\documentclass[a4paper,11pt]{article}
\usepackage[utf8]{inputenc}
\usepackage{amsfonts,amsmath,amsthm}
\usepackage{graphicx}
\usepackage[margin=3cm]{geometry}

\newtheorem{thm}{Theorem}
\newtheorem{lem}[thm]{Lemma}
\newcommand{\eps}{\varepsilon}
\newcommand{\diam}{\mathrm{diam}}
\newcommand{\ud}[1]{\bigcirc_{#1}}
\newcommand{\OPT}{\mathit{OPT}}
\newcommand{\E}{\mathrm{E}}
\newcommand{\CC}{\mathcal{C}}

\begin{document}
\title{Shifting coresets: obtaining linear-time approximations for \\unit disk graphs and other geometric intersection graphs}

\author{%
	Guilherme D. da Fonseca\\
		IUT and LIMOS\\
		Universit\'e d'Auvergne \\
		fonseca@isima.fr\\
	\and
	Vin\'icius Gusm\~ao Pereira de S\'a\\
		Instituto de Matem\'atica\\
		Universidade Federal do Rio de Janeiro \\
		vigusmao@dcc.ufrj.br\\
	\and
	Celina Miraglia Herrera de Figueiredo\\
		COPPE, Universidade Federal do Rio de Janeiro \\
		celina@cos.ufrj.br
}

\date{To appear at International Journal of Computational Geometry and Applications}

\maketitle

\begin{abstract}
Numerous approximation algorithms for problems on unit disk graphs have been proposed in the literature, exhibiting a sharp trade-off between running times and approximation ratios.
We introduce a variation of the known shifting strategy that allows us to obtain linear-time constant-factor approximation algorithms for such problems. To illustrate the applicability of the proposed variation, we obtain results for three well-known optimization problems.
Among such results, the proposed method yields linear-time $(4+\eps)$-{ap\-prox\-i\-ma\-tions} for the maximum-weight independent set and the minimum dominating set of unit disk graphs, thus bringing significant performance improvements when compared to previous algorithms that achieve the same approximation ratios. 
Finally, we use axis-aligned rectangles to illustrate that the same method may be used to derive linear-time approximations for problems on other geometric intersection graph classes.
\end{abstract}

\section{Introduction} \label{s:introduction}

Linear- and near-linear-time approximation algorithms constitute an active topic of research, even for problems that can be solved exactly in polynomial time, such as maximum flow and maximum matching~\cite{agarwal2014,maxflow,matching14,maxflow14,matchings}. This paper employs a variation of the \emph{shifting strategy} introduced by Hochbaum and Maass three decades ago~\cite{shifting}. While the shifting strategy usually leads to high running times, our variation, which we call the \emph{shifting coresets method}, allows us to obtain linear-time constant-approximation algorithms for problems on unit disk graphs.
Three algorithms for well-studied optimization problems on unit disk graphs are presented to illustrate the method. A fourth algorithm, for maximum-weight independent sets of axis-aligned rectangles in the plane, is also given, showing that the method applies to other geometric intersection graph classes as well.

\emph{Geometric intersection graphs} are graphs whose $n$ vertices correspond to geometric objects and whose $m$ edges correspond to pairs of intersecting objects. Several classes of geometric intersection graphs are defined by restricting the shape of the objects: disks, unit disks, squares, rectangles, etc. Approximation algorithms for such classes are either \emph{graph-based}, when they receive as input solely the adjacency representation of the graph, or \emph{geometric}, when the input consists of a geometric description of the objects. Note that when the goal is to design $O(n)$-time algorithms, the geometric representation is required, since the number $m$ of edges in a geometric intersection graph can be as high as~$\Theta(n^2)$.

\emph{Polynomial-time approximation schemes} (PTASs) have been developed for several optimization problems on multiple classes of intersection graphs. Among such problems, we focus on the following three in the present paper:

\begin{itemize}
\item \emph{Maximum-weight independent set} (WIS): Given a graph with weighted vertices, find the maximum-weight subset of mutually non-adjacent vertices.

\item \emph{Minimum dominating set} (DS): Given a graph, find the minimum-cardinality vertex subset $D$ such that every vertex not in $D$ is adjacent to some vertex in $D$.

\item \emph{Minimum vertex cover} (VC): Given a graph, find the minimum-cardinality vertex subset $C$ such that every edge is incident on some vertex in $C$.
\end{itemize}

The shifting strategy has given rise to geometric PTASs for several problems on geometric intersection graphs~\cite{shifting,Aga98,ptas-mcds,ptas-geometric}. A set of $n$ geometric objects has \emph{constant diameter} if the Euclidean distance between any two points contained inside the objects is upper-bounded by a constant.
Essentially, the shifting strategy reduces the original problem to a set of subproblems of constant diameter. Such a reduction takes $O(n)$ time and yields a $(1+\eps)$-{ap\-prox\-i\-ma\-tion} to the original problem once the solutions to the subproblems are known.
Each subproblem is then solved exactly by exploiting the fact that it has constant diameter. For example, it is often possible to show by packing arguments that an input instance whose diameter is $d$ admits a solution with $c = O(d^2)$ vertices, so that exhaustive enumeration can find the optimal solution in roughly $O(n^c)$ time. As a consequence, the running times of such PTASs are high-degree polynomials. 
Other PTASs that are not based on the shifting strategy also exist, but their complexities are usually even higher~\cite{ptas-graph-journal}.

Among countless classes of geometric intersection graphs, unit disk graphs are arguably the most studied one, owing largely to their applicability in wireless networks~\cite{ptas-graph-journal,heuristics}. A \emph{unit disk graph} (UDG) is the intersection graph of unit disks in the plane, and its vertices are usually represented by the coordinates of the disk centers. With respect to this representation, two vertices are adjacent if the corresponding points (the disk centers) are within Euclidean distance at most $2$ from one another. Next, we review the state of the art of the three aforementioned optimization problems on the class of UDGs.

The minimum dominating set problem (DS) on UDGs admits some PTASs, the fastest of which is geometric and provides a $(1+\eps)$-approximation in $n^{O(1/\eps^2)}$ time~\cite{ptas-geometric,ptas-graph-journal}. For the sake of comparison with the linear-time $(4+\eps)$-approximation algorithm introduced in Section~\ref{s:DS}, such a PTAS produces a $4$-approximation in roughly $O(n^{10})$ time. The high running times of the existing PTASs have therefore motivated the study of faster constant-factor approximation algorithms.
Examples of graph-based algorithms include a $44/9$-{ap\-prox\-i\-ma\-tion} that runs in $O(n + m)$ time and a $43/9$-{ap\-prox\-i\-ma\-tion} that runs in $O(n^2m)$ time~\cite{linear-tcs}. Among the geometric algorithms, we cite
the original $5$-{ap\-prox\-i\-ma\-tion}, which can be implemented in $O(n)$ time if the floor function and constant-time hashing are available~\cite{heuristics}; 
a $44/9$-{ap\-prox\-i\-ma\-tion} that uses local improvements and runs in $O(n \log n)$ time~\cite{linear-tcs};
a $4$-{ap\-prox\-i\-ma\-tion} that uses grids and runs in $O(n^8 \log n)$ time~\cite{cccg-journal};
and an improved $4$-{ap\-prox\-i\-ma\-tion} that uses hexagonal grids and runs in $O(n^6 \log n)$ time~\cite{hex-journal}.

The maximum-weight independent set problem (WIS) also admits PTASs for UDGs, the fastest of which is geometric and attains a $(1+\eps)$-approximation in $O(n^{4\lceil2/\eps\sqrt{3}\rceil})$ time~\cite{ptas-geometric,ptas-graph-journal,mwis-ptas}. 
Using a slab decomposition and the longest path in a DAG, Matsui~\cite{mwis-ptas} showed how to obtain an approximation ratio of $(1+(2/\sqrt{3})+\eps) < 2.16$ in $O(n^2)$ time.\footnote{After this article, the running time has been reduced to $O(n \log^3 n)$ using advanced data structures~\cite{is2apx2}.} Alternatively, a geometric $5$-{ap\-prox\-i\-ma\-tion} can be obtained in $O(n \log n)$ time by a greedy approach that considers the vertices in decreasing order of weights, or in $O(n+m)$ time in the graph-based version~\cite{heuristics}. In comparison, the algorithm presented in Section~\ref{s:WIS} obtains a $(4+\eps)$-approximation in linear time. 
Note that, for the unweighted version, a geometric greedy approach that considers the vertices from left to right can be implemented to give a $3$-{ap\-prox\-i\-ma\-tion} in $O(n)$ time with floor function and constant-time hashing~\cite{heuristics}. The high running time of the existing PTASs also motivated the discovery of an $O(n^3)$-time $2$-{ap\-prox\-i\-ma\-tion} algorithm for the unweighted version~\cite{is2apx}.

The minimum vertex cover problem (VC) for UDGs is a rare example for which a geometric linear-time approximation scheme is known~\cite{marx05}. Previously, two high-complexity PTASs (one geometric and one graph-based) had been proposed~\cite{ptas-geometric,vc-graph}.

Another important class of geometric intersection graphs is defined by axis-aligned rectangles in the plane. Maximum-weight independent sets in such graphs are widely studied due to their application in data mining, map labeling, and VLSI design~\cite{Aga98,Aga06,Chan03,Chan04,Adamaszek13}. However, the lack of polynomial constant-factor approximation algorithms has motivated the investigation of more restricted subclasses, for example the intersection graphs of squares and of unit-height rectangles~\cite{Aga98,Chan04}. Also, PTASs for the WIS exist for the class of fat convex objects, which generalizes unit disk graphs and several intersection graphs of rectangles~\cite{Chan03,Seidel}. The fastest PTAS for the WIS on fat convex objects takes $n^{O(1/\eps)}$ time and space, or alternatively $n^{O(1/\eps^2)}$ time with $O(n)$ space~\cite{Chan03}.

\paragraph{Our results.}
We present a method to solve the constant-diameter subproblems that appear when the shifting strategy is used. We call it the \emph{shifting coresets method} (Section~\ref{s:ourmethod}). By incorporating this method into the shifting strategy idea, we obtain linear-time approximation algorithms for problems on unit disk graphs and other geometric intersection graphs.
The method is based on approximating the input set of geometric objects, which can be arbitrarily dense, by a \emph{sparse} set of objects, that is, a set of objects such that any square of constant size contains at most a constant number of objects.

To obtain efficient algorithms through the application of our method, one needs to investigate the fundamental question of how well a sparse set---generated using only local information---can approximate a denser set for each considered problem.
Thus, although the algorithms in this text share the same basic strategy, their analyses differ significantly. For example, the WIS analyses apply advanced graph-coloring theorems, whereas the DS analysis applies geometric packing arguments.

\begin{table}
\begin{center}
\begin{small}

\begin{tabular}{l|ll|ll} 
previous / new results &WIS  & &DS & \\
\hline
\rule{0pt}{3ex}previous ratio in $O(n\, \textrm{polylog}(n))$ time &$5$ & \cite{heuristics} &$4{.}889$ & \cite{linear-tcs} \\
our approximation ratio in $O(n)$ time &$4+\eps$ & Sec. \ref{s:WIS} &$4+\eps$ & Sec. \ref{s:DS}\\
previous time for a ($4+\eps$)-approximation &$O(n^2)$ & \cite{mwis-ptas} &$O(n^6 \log n)$\hspace{-.2cm} & \cite{hex-journal} 
\end{tabular}

\end{small}
\end{center}
\caption{\label{t:table}Comparison of new and previous approximation algorithms for UDGs.}
\end{table}

By using the shifting coresets method, we obtain linear-time $(4+\eps)$-{ap\-prox\-i\-ma\-tion} algorithms on UDGs for the WIS (which is tackled first, in Section~\ref{s:WIS}, owing to its greater simplicity) and for the DS (Section~\ref{s:DS}). 
The proposed algorithms provide significant improvements when compared to existing linear- and near-linear-time algorithms for these problems (see Table~\ref{t:table}).
The text also includes (Section~\ref{s:VC}) an alternative linear-time $(1+\eps)$-{ap\-prox\-i\-ma\-tion} obtained for the VC. We reiterate that in that section we have independently derived a previous result from Marx, for the sake of illustrating an indirect application of our method~\cite{marx05}. 
Finally, our method is applied to obtain a linear-time $(6+\eps)$-approximation algorithm for the WIS on the class of intersection graphs of axis-aligned rectangles with side lengths in $[1, \lambda]$, for an arbitrary constant $\lambda \geq 1$.

In Section~\ref{s:conclusion}, some open problems and lower bounds to the approximation ratios of the presented algorithms are discussed.

\section{The shifting coresets method}\label{s:ourmethod}

The shifting strategy is the core of most of the existing geometric PTASs for problems on geometric intersection graphs~\cite{shifting}.
Generally, the shifting strategy reduces the original problem with $n$ objects to a set of subproblems whose inputs have constant diameter and the sum of the input sizes is $O(n)$. Such a reduction is based on partitioning the objects according to a number of iteratively shifted grids and takes $O(n)$ time (by using the floor function and constant-time hashing). 
Exploiting the inputs' constant diameter, each subproblem is solved exactly in polynomial time. The solutions to the subproblems are then combined  appropriately (normally in $O(n)$ time) to yield feasible solutions to the original problem, the best of which is returned. The high complexities of these geometric PTASs are due to the exact algorithms that are employed to solve each subproblem.

The shifting coresets method we introduce is based on the shifting strategy. However, it presents a crucial new aspect. Rather than obtaining exact, costly solutions for the subproblems, it solves each subproblem \emph{approximately}. To do that, it employs the coresets paradigm, which consists in considering only a constant-size subset of the input~\cite{coresets}.
~
Let $\diam(P)$ denote the maximum distance between two points contained inside objects of a set of objects $P$. For a problem whose input is a set $P$ of $n$ objects and a parameter $\eps > 0$, the method can be briefly described as follows. 
\begin{enumerate}
\item Apply the shifting strategy to construct a set of $r$ subproblems with inputs $P_1,\ldots,P_r$ such that $\sum_{i=1}^r |P_i| = O(n)$ and for all $i$, $\diam(P_i)$ is a constant that depends on $\eps$.
\item For each subproblem instance $P_i$, obtain a coreset $Q_i \subseteq P_i$ with $|Q_i| = O(1)$, such that the optimal solution for instance $Q_i$ is an $\alpha$-{ap\-prox\-i\-ma\-tion} to the optimal solution for instance $P_i$. 
\item Solve the problem exactly for each $Q_i$.
\item Combine the solutions into an $(\alpha + \eps)$-{ap\-prox\-i\-ma\-tion} for the original problem.
\end{enumerate}

Coresets for different problems must be devised on a case-by-case basis.
For the WIS on UDGs, we create a grid with cells of diameter $0{.}29$ and consider only one disk center of maximum weight inside each cell. 
For the DS on UDGs, we create a grid with cells of diameter $0{.}24$ and consider only the (at most four) disk centers, inside each cell, with minimum or maximum coordinate in either dimension (breaking ties arbitrarily). 
We solve the VC on UDGs by breaking each subproblem into two cases. In the first one, the number of input disks is already bounded by a constant, therefore already a coreset. In the second one, we use the same coreset as in the WIS. Finally, we solve the WIS on a class of rectangle intersection graphs by representing each rectangle as a $4$-dimensional point and using a $4$-dimensional grid with cells of side $0{.}1$. Our coreset is then formed by the rectangle of maximum weight inside each $4$-dimensional cell.

Assuming a real-RAM model of computation with floor function and constant-time hashing, it is possible to partition the input points into grid cells efficiently, yielding an overall $O(n)$ running time for the algorithms~\cite{fixed-radius}. Without these operations, the running time becomes $O(n \log n)$.
We also assume that $\eps$ is constant. Otherwise, the running time becomes $2^{O(1/\eps^2)}n$ for the WIS and the DS on UDGs, since their corresponding coresets contain $O(1/\eps^2)$ points.\footnote{The dependency on $\eps$ may be conceivably improved to $2^{O(1/\eps)}$---for instance, by replacing brute-force search with a $2^{O(\sqrt{n})}$-time exact algorithm via separators. Such a dependency is not too important, though, since the approximation factor does not approach $1$. We have therefore opted for the simplest approach.} 
For the VC, the running time becomes $2^{O(1/\eps^3)}n$. As for the WIS on axis-aligned rectangles, if one regards $\eps$ and $\lambda$ (the maximum allowed side length) as asymptotic variables, then the running time becomes $2^{O(\lambda^2/\eps^2)}n$.

\smallskip

\paragraph{Definitions.}
A grid is said to be \emph{rooted at} a point $(x,y)$ if there is a grid cell with corner at $(x,y)$.
Given a grid cell $C$ and a real parameter $d>0$, the square region $C_d^-\subset C$, called the \emph{$d$-contraction} of $C$, is formed by removing from $C$ the points within distance at most $d$ from the boundary of $C$. Analogously, the square region~$C_d^+ \supset C$, called the \emph{expansion} of $C$, is formed by $C$ and all points within $L_\infty$ distance at most $d$ from $C$. Finally, we denote the collection of all $d$-expansions of cells of a grid rooted at a point $p$ by $\CC_d^+(p)$.

\section{WIS on UDGs} \label{s:WIS}

In this section, we show how to apply the shifting coresets method to obtain a linear-time $(4+\eps)$-{ap\-prox\-i\-ma\-tion} to the WIS on unit disk graphs. We start by presenting a $4$-{ap\-prox\-i\-ma\-tion} for point sets of constant diameter, and then we use the shifting strategy to obtain the desired $(4+\eps)$-{ap\-prox\-i\-ma\-tion}.

Given a point $p$ and a set $S$ of points, let $w(p)$ denote the weight of $p$, and let $w(S)=\sum_{p \in S}w(p)$. Two or more points are considered \emph{independent} if their minimum distance is strictly greater than $2$.

\begin{thm} \label{t:constdiam-WIS}
Given a set $P$ of $n$ points with real weights as input, with $\diam(P) = O(1)$, the WIS for the corresponding UDG can be $4$-approximated in $O(n)$ time on the real-RAM.
\end{thm}
\begin{proof}
Our algorithm proceeds as follows. 
First, find the points of $P$ with minimum or maximum coordinates in either dimension. 
That defines a bounding box of constant size for $P$. Within this bounding box, create a grid with cells of diameter $\gamma = 0{.}29$ and side $\gamma / \sqrt{2}$ (any value $\gamma < (2-\sqrt{2})/2$ suffices). Note that the number of grid cells is constant, and therefore the points of $P$ can be partitioned among the grid cells in $O(n)$ time (even without using the floor function or hashing). Then, build the subset $Q \subseteq P$ as follows. For each non-empty grid cell $C$, add to $Q$ a point of maximum weight in $P \cap C$. Afterwards, determine the maximum-weight independent set $I^*$ of $Q$. Since $|Q| = O(1)$, this can be done in constant time. Return the solution $I^*$.

Next, we show that $I^*$ is indeed a $4$-{ap\-prox\-i\-ma\-tion}. We argue that, given an independent set $I \subseteq P$, there is an independent set $I' \subseteq Q$ with $4~w(I') \geq w(I)$. Given a point $p \in P$, let $q(p)$ denote the point from $Q$ that is contained in the same grid cell as $p$. Consider the set $S = \{q(p) : p \in I \}$. Note that $w(q(p)) \geq w(p)$ and $w(S) \geq w(I)$. The set $S$ may not be independent, but since $I$ is independent, the minimum distance in $S$ is at least $2-2\gamma = 1{.}42 > \sqrt{2}$. We claim that the unit disk graph formed by $S$ is a planar graph. To prove the claim, we show that a planar drawing can be obtained by connecting the points of $S$ within distance at most $2$ by straight line segments. Given a pair of points $p_1,p_2$ with distance $\|p_1p_2\| \leq 2$, the Pythagorean Theorem shows that a unit disk centered within distance greater than $\sqrt{2}$ from both $p_1$ and $p_2$ cannot intersect the segment $p_1p_2$. 
This implies no two edges in the drawing can cross.
By the Four-Color Theorem~\cite{fourcolor}, $S$ admits a partition into four independent sets $S_1,\ldots,S_4$. The set $I'$ of maximum weight among $S_1,\ldots,S_4$ must have weight at least $w(I)/4$.\footnote{Note that the Four-Color Theorem is only used in the argument. No coloring is ever computed by the algorithm.}

Since $I^*$ is the maximum-weight independent set of $Q$, it follows that $I^*$ is a 4-{ap\-prox\-i\-ma\-tion} for the WIS. 
\end{proof}

The following theorem uses the shifting strategy to obtain a $(4+\eps)$-{ap\-prox\-i\-ma\-tion} for point sets of arbitrary diameter. The proof borrows ideas from {Hunt III} \emph{et al}~\cite{ptas-geometric}. We present them in a different manner, though, including details about an efficient implementation of the strategy.

\begin{thm} \label{t:WIS}
Given a set $P$ of $n$ points in the plane as input, the WIS for the corresponding UDG can be $(4+\eps)$-approximated in $O(n)$ time on the real-RAM with constant-time hashing and the floor function. Without these operations, it can be done in $O(n \log n)$ time.
\end{thm}
\begin{proof}
Let $k$ be the smallest integer such that 
\begin{equation}\label{eq:k1}
\left(\frac{k-2}{k}\right)^2 \geq \frac{4}{4+\eps}.
\end{equation} 

The algorithm proceeds as follows. For $i,j$ from $0$ to $k-1$, create a grid with cells of side $2k$ rooted at $(2i,2j)$. For each cell $C$ in the grid, run the WIS $4$-{ap\-prox\-i\-ma\-tion} algorithm from Theorem~\ref{t:constdiam-WIS} with point set $P \cap C_2^-$ (the points in $P$ that belong to a $2$-contraction of $C$, see Fig.\ref{f:contraction}), obtaining a solution $I_{i,j}(C)$. Then, the independent set $I_{i,j}$ is constructed as the union of the independent sets $I_{i,j}(C)$ for all grid cells $C$. Return the maximum-weight set $I_{i,j}$ that is found, call it $I^*$.

 \begin{figure}[t]
  \centering
  \includegraphics[scale=.45]{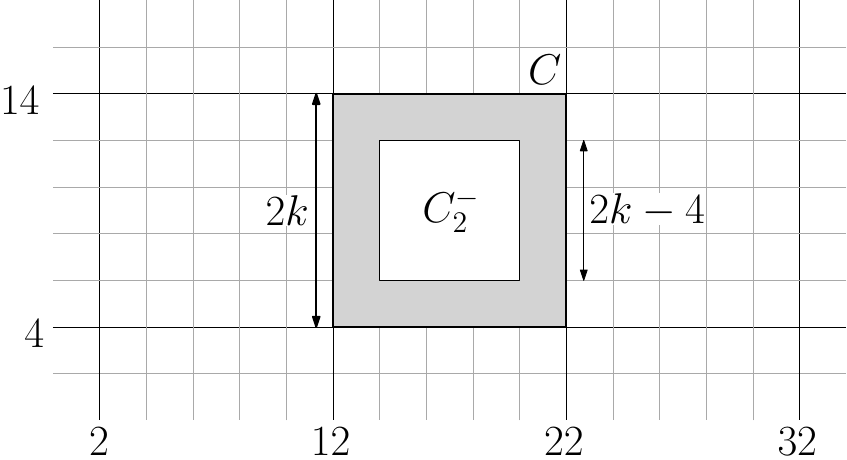}
  \caption{Grid rooted at $(2,4)$ with cells of side $2k=10$ and the $2$-contraction of a cell.}
  \label{f:contraction}
 \end{figure}

To implement the algorithm efficiently, create a subgrid of subcells of side $2$, assigning each point to the subcell that contains it. In order to partition the $n$ points into subcells efficiently, one should use the floor function and constant-time hashing, taking $O(n)$ time. If these operations are not available, it is possible to determine the connected components of the graph (using the Delaunay triangulation, for example) and to partition the points of each component into subcells (through sorting by $x$-coordinate, separating into columns, and sorting inside each column by $y$-coordinate). The non-empty subcells are stored in a balanced binary search tree. This process takes $O(n \log n)$ time due to sorting, Delaunay triangulation, and binary search tree operations. Given the partitioning of the point set into subcells, each input to the WIS algorithm can be constructed as the union of a constant number of subcells. Finally, the total size of the constant-diameter WIS instances is $O(n)$, since each point from the original point sets 
appears in a constant number of such instances, namely $(k-2)^2$.

To prove that the returned solution $I^*$ is indeed a $(4+\eps)$-{ap\-prox\-i\-ma\-tion}, we use a probabilistic argument. Let $i,j$ be picked uniformly at random from $0,\ldots,k-1$ and let $\OPT$ denote the optimal solution. For every cell $C$, we have
$$w(I_{i,j}(C)) \geq \frac{1}{4}~w(\OPT \cap C_2^-).$$
Consequently, by summing over all grid cells,
$$w(I_{i,j}) = \sum_C w(I_{i,j}(C)) \geq \frac{1}{4} \sum_C w(\OPT \cap C_2^-)
= \frac{1}{4} \sum_{p \in \OPT}\rho(p)~w(p),$$
where $\rho(p)$ denotes the probability that a given point $p$ is contained in some contracted cell. Because such probability is the same for all points, $\E[w(I_{i,j})]$ is bounded by 
$$
\E[w(I_{i,j})] \geq \frac{1}{4}~\rho(p)  \sum_{p \in \OPT}w(p)
                 = \frac{1}{4}~\rho(p)~w(\OPT).
$$
Note that, for all $p \in P$,
$\rho(p)$ corresponds to the ratio between the areas of $C_2^-$ and $C$, namely 
$$
\rho(p) = \frac{\mathrm{area}(C_2^-)}{\mathrm{area}(C)} = \left(\frac{k-2}{k}\right)^2.
$$
Therefore, by using inequality~(\ref{eq:k1}), we obtain
$$
\E[w(I_{i,j})] \geq \frac{1}{4}\left(\frac{k-2}{k}\right)^2w(\OPT) \geq \frac{1}{4+\eps}~w(\OPT).
$$

Since $I^*$ has maximum weight among the independent sets $I_{i,j}$, it follows that $w(I^*)$ is at least as large as their average weight. Therefore, 
$I^*$ 
satisfies 
$$
w(I^*) \geq \E[w(I_{i,j})] \geq \frac{1}{4+\eps}~w(\OPT),
$$ 
closing the proof. 
\end{proof}

\section{DS for UDGs} \label{s:DS}
In this section, we show how to apply the shifting coresets method to obtain a linear-time $(4+\eps)$-{ap\-prox\-i\-ma\-tion} to the DS (in fact, a generalization of it) on unit disk graphs. We start by presenting a $4$-{ap\-prox\-i\-ma\-tion} for point sets of constant diameter, and then we use the shifting strategy to obtain the desired $(4+\eps)$-{ap\-prox\-i\-ma\-tion}. A point $p$ dominates a point $q$ if $\|pq\| \leq 2$.
Given two sets of points $D$ and $P'$, the set $D$ is a \emph{$P'$-dominating set} if every point in $P'$ is dominated by some point in $D$.

We now define a more general version of the DS, the \emph{minimum partial dominating set problem} (PDS). Such a generalization is necessary to properly apply the shifting strategy. In the PDS, given a set $P$ of $n$ points and also a subset $P' \subseteq P$, the goal is to find the smallest $P'$-dominating subset $D \subseteq P$.

In order to analyze our algorithm, we prove a geometric lemma that shows that the difference between a unit circle and two unit disks that are sufficiently close to it and form a sufficiently big angle consists of one or two ``small'' arcs. Given a point $p$, let $\ud{p}$ denote the unit disk centered at $p$, and $\partial{\ud{p}}$ denote its boundary circle.

\begin{lem} \label{l:1or2arcs}
Given $\delta > 0$ and three points $p,q_1,q_2 \in \mathbb{R}^2$ with (i)~$\|pq_1\| \leq \delta$, (ii)~$\|pq_2\|	 \leq \delta$, and (iii)~the smallest angle $\angle q_1pq_2$ is greater than or equal to $\pi/2$, it follows that:
\begin{enumerate}
 \item[(1)] the portion $T = (\partial\ud{p}) \setminus (\ud{q_1} \cup \ud{q_2})$ of the boundary $\partial\ud{p}$ consists of one or two circular arcs;
 
 \item[(2)] if $T$ consists of one circular arc, then the arc length is less than or equal to $\pi/2 + 2\arcsin(\delta/2)$; and 
 
 \item[(3)] if $T$ consists of two circular arcs, then each arc length is less than $2 \arcsin \delta$. 
\end{enumerate}
\end{lem}
\begin{proof}
Statement (1) is clearly true. We start by proving statement (2). The arc length $\|T\|$ is maximized as the angle $\angle q_1pq_2$ decreases while the distances $\|pq_1\|,\|pq_2\|$ are kept constant, therefore it suffices to consider the case in which $\angle q_1pq_2 = \pi / 2$. The arc $T$ centered at $p$ can be decomposed into three arcs by rays in directions $q_1p$ and $q_2p$, as shown in Fig.~\ref{f:1arc}(a). The central arc measures $\pi/2$, while each of the other two arcs measures $\arcsin(\delta/2)$, proving statement~(2).

\begin{figure}[t]
 \centering 
 \includegraphics[scale=.45]{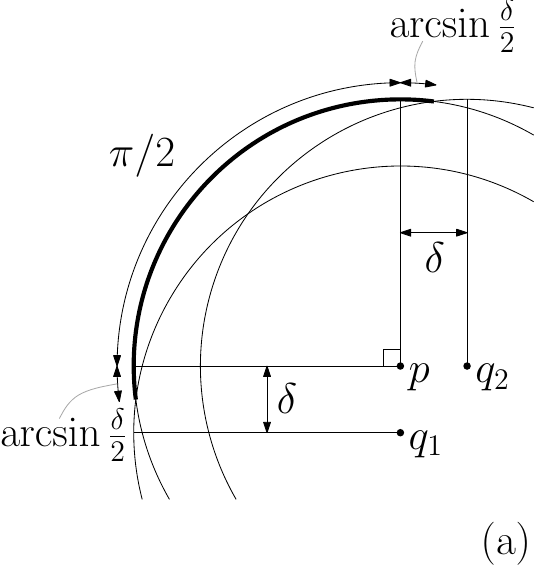} \hspace{.5cm}   
 \includegraphics[scale=.45]{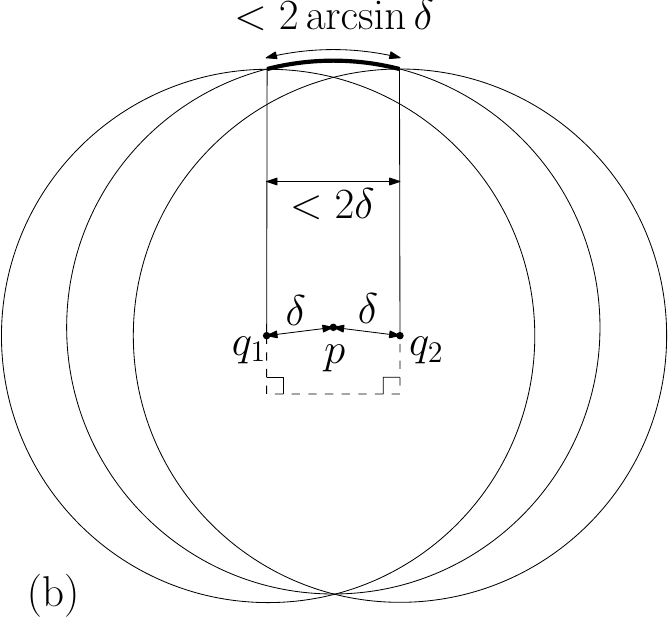}
 \caption{Proof of Lemma~\ref{l:1or2arcs}.}
 \label{f:1arc} \label{f:2arcs}
\end{figure}

Next, we prove statement (3). Let $T_1,T_2$ denote the two arcs that form $T$ with $\|T_1\| \geq \|T_2\|$. The arc length $\|T_1\|$ is maximized in the limit when $\|T_2\| = 0$, as shown in Fig.~\ref{f:2arcs}(b). The rays connecting $q_1$ and $q_2$ to the two extremes of $T_1$ are parallel, and therefore $\|T_1\| < 2 \arcsin \delta$. 
\end{proof}

We are now able to prove the following theorem, which presents our $4$-{ap\-prox\-i\-ma\-tion} algorithm for point sets of constant diameter.

\begin{thm} \label{t:constdiam-PDS}
Given two sets of points $P$ and $P'$ as input, with $P' \subseteq P$, $|P|=n$, and $\diam(P) = O(1)$, the PDS can be $4$-approximated in $O(n)$ time on the real-RAM.
\end{thm}
\begin{proof}
First, determine a bounding box of constant size for $P$, as in the algorithm for the WIS.
Within this bounding box, create a grid with cells of diameter $\gamma = 0{.}24$ and side $\gamma / \sqrt{2}$
(any positive $\gamma$ satisfying 
$$\gamma + \sqrt{8 - 8\cos\left(\frac{\frac{\pi}{2} + 2 \arcsin(\frac{\gamma}{2}) }{2}\right)} < 2$$
suffices).
Note that the number of grid cells is constant, and therefore the points of $P$ can be partitioned among the grid cells in $O(n)$ time (even without using the floor function or hashing). Then, build the subset $Q \subseteq P$ as follows. 
For each non-empty grid cell, add to $Q$ the (at most four) extreme points inside the cell, i.e., those presenting minimum or maximum coordinate in either dimension. Ties are broken arbitrarily.
Since there is a constant number of grid cells and $Q$ contains at most four points per cell, it follows that $|Q| = O(1)$. Now determine the smallest $P'$-dominating subset $D^* \subseteq Q$. To do that, examine the subsets of $Q$, from smallest to largest, verifying if all points of $P'$ are dominated, until the dominating set $D^*$ if found and returned as the approximate solution. Since $Q$ has a constant number of points, this procedure takes $O(n)$ time.

Now we show that the returned solution $D^*$ is indeed a $4$-{ap\-prox\-i\-ma\-tion}. 
We argue that, given a $P'$-dominating set $D \subseteq P$, there is a $P'$-dominating set $D' \subseteq Q$ with $|D'| \leq 4~|D|$. To build the set $D'$ from $D$, proceed as follows. For each point $p \in D$, if $p \in Q$, add $p$ to $D'$. Otherwise, since the set $Q$ contains 
points of extreme coordinates in both $x$ and $y$ axes,
in the cell of~$p$, there are two points $q_1,q_2 \in Q$ such that (i) $\|pq_1\| \leq \gamma$, (ii) $\|pq_2\| \leq \gamma$, and (iii)~the smallest angle $\angle q_1pq_2$ is at least $\pi/2$. Add these two points $q_1,q_2$ to~$D'$.

By Lemma~\ref{l:1or2arcs}, the portion $T = (\partial\ud{p}) \setminus (\ud{q_1} \cup \ud{q_2})$ of $\partial\ud{p}$ consists of one or two circular arcs. First consider the case in which $T$ consists of one circular arc.
Let $R$ be the set of points from $P'$ which are dominated by $p$, but not by $q_1$ or~$q_2$. If $R$ is empty, then no extra point needs to be added to $D'$. Otherwise, the line $\ell$ which contains $p$ and bisects $T$ separates $R$ into two (possibly empty) sets $R_1,R_2$. If $R_1 \neq \emptyset$, let $p_3$ be an arbitrary point of $R_1$. Since $Q$ contains a point in the same cell as $p_3$, there is a point $q_3$ with $\|p_3q_3\| \leq \gamma$. Add the point $q_3$ to $D'$. Analogously, if $R_2 \neq \emptyset$, let $p_4$ be an arbitrary point of $R_2$ and let $q_4 \in Q$ be a point with $\|p_4q_4\| \leq \gamma$. Add the point $q_4$ to $D'$.

We now show that the four points $q_1,q_2,q_3,q_4 \in Q$ dominate all points dominated by $p$. Consider a point $v$ that is dominated by $p$ but not by $q_1$ or $q_2$. The point $v$ must be inside the circular crown sector depicted in Fig.~\ref{f:crownsector}(a) and described as follows. Because $v$ is dominated by $p$, we have $\|pv\|\leq 2$. By Lemma~\ref{l:1or2arcs}, the arc length $\|T\| < 1{.}82$. Also, $\|pv\| > 1$, because otherwise the unit circles centered at $p$ and $v$ would intersect forming an arc of length at least $2\pi/3$, which is greater than $\|T\|$, in which case $v$ is dominated by $q_1$ or $q_2$. 
Finally, since $v$ is closer to $p$ than it is to $q_1$ or $q_2$, it follows that $v$ must be between the lines that connect $p$ to the endpoints of~$T$. This circular crown sector is bisected by the line~$\ell$. By the law of cosines, the diameter of each circular crown sector is $d = \sqrt{8 - 8\cos(\|T\|/2)} < 1{.}76$. Therefore, for any point $v$ inside the circular crown sector, the point $q_3$ (or $q_4$, analogously) that is within distance at most $\gamma$ from a point inside the same sector dominates $v$, as $\|vq_3\| \leq d + \gamma < 2$.

\begin{figure}[t]
 \centering
 \includegraphics[scale=.45]{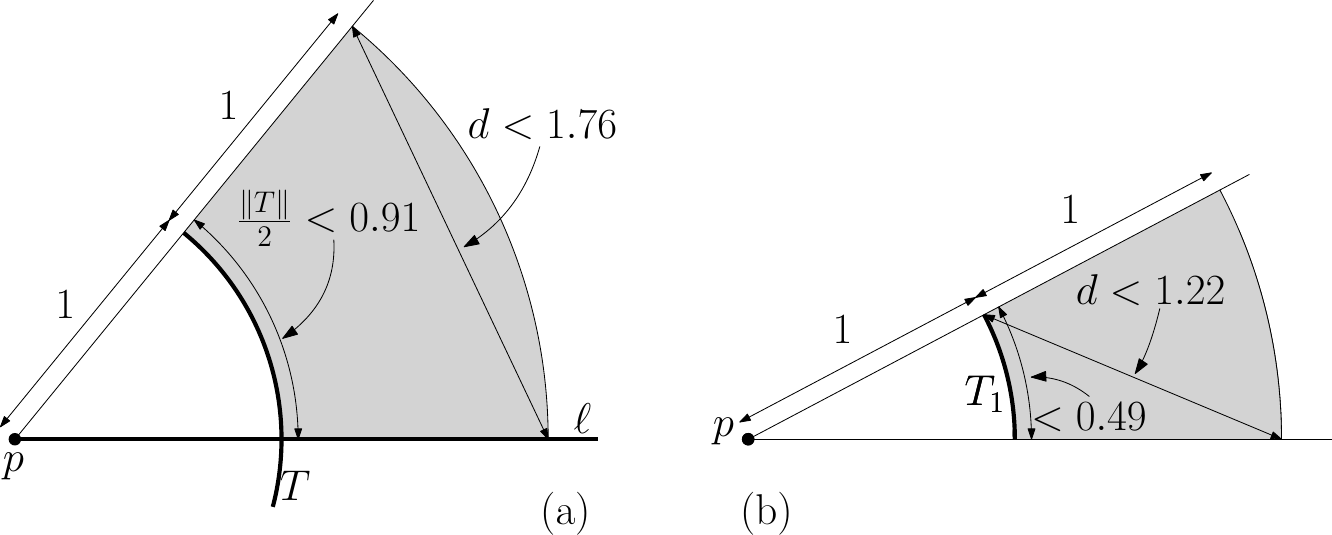}
 \caption{Proof of Theorem~\ref{t:constdiam-PDS}.}
 \label{f:crownsector}
\end{figure}

Finally, if $T$ consists of two circular arcs $T_1,T_2$ centered in $p$, then start by adding those same points $q_1, q_2$ to $D'$, as if $T$ consisted of only one arc. Then, if necessary, add new points $q_3,q_4$ to $D'$ as follows.
The points that are dominated by $p$ but not by $q_1$ or $q_2$ must be within distance $1$ of either $T_1$ or $T_2$. Let $p_3, p_4$ be arbitrary points that are within distance $1$ of $T_1$ or $T_2$, respectively, but are not dominated by $q_1$ or $q_2$. 
If such points $p_3,p_4$ exist, then there are two points $q_3,q_4$ in $Q$ that are within distance at most $\gamma$ from respectively $p_3, p_4$. By Lemma~\ref{l:1or2arcs}, the largest arc among $T_1,T_2$ measures at most $0{.}49$. The proof that all points dominated by $p$ are dominated by $q_1,q_2,q_3$, or $q_4$ is analogous to the case in which $T$ consists of a single arc, using the circular crown sector illustrated in Fig.~\ref{f:crownsector}(b).

Since $D^*$ is minimum among all subsets of $Q$ that are $P'$-dominating sets, $D^*$ is a 4-{ap\-prox\-i\-ma\-tion} for the PDS. 
\end{proof}

The following theorem uses the shifting strategy to obtain a $(4+\eps)$-{ap\-prox\-i\-ma\-tion} for point sets of arbitrary diameter.

\begin{thm} \label{t:PDS}
Given two sets of points $P$ and $P'$ as input, with $P' \subseteq P$ and $|P|=n$, the PDS can be $(4+\eps)$-approximated in $O(n)$ time on the real-RAM with constant-time hashing and the floor function. Without these operations, it can be done in $O(n \log n)$ time.
\end{thm}
\begin{proof}
Let $k$ be the smallest integer such that 
\begin{equation}\label{eq:k2}
\left(\frac{k+2}{k}\right)^2 \leq 1+\frac{\eps}{4}.
\end{equation} 

The algorithm proceeds as follows. For $i,j$ from $0$ to $k-1$, create a grid with cells of side $2k$ rooted at $(2i,2j)$ and, for each cell $C$ in the grid, use Theorem~\ref{t:constdiam-PDS} to $4$-approximate the PDS with point sets $P \cap C_2^+$ (the points of $P$ that belong to the $2$-expansion of $C$, see Fig.~\ref{f:expansion}) and $P' \cap C$, obtaining a solution $D_{i,j}(C)$. The dominating set $D_{i,j}$ is constructed as the union of the dominating sets $D_{i,j}(C)$ for all grid cells $C$. Return the smallest dominating set $D_{i,j}$ that is found, call it $D^*$.

\begin{figure}
 \centering
 \includegraphics[scale=.45]{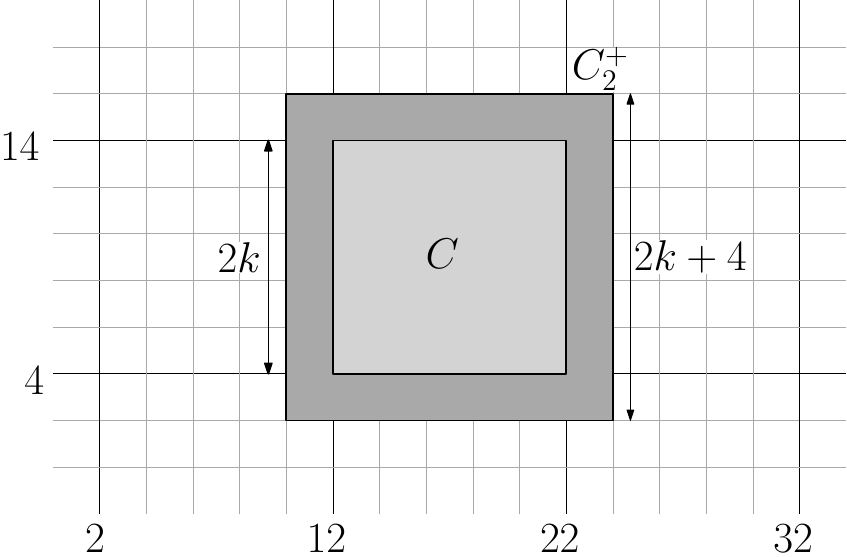}
 \caption{Grid rooted at $(2,4)$ with cells of side $10$ and the $2$-expansion of a cell.}
 \label{f:expansion}
\end{figure}

To prove that the returned solution is indeed a $(4+\eps)$-{ap\-prox\-i\-ma\-tion}, we use again a probabilistic argument. Let $i,j$ be picked uniformly at random from $0,\ldots,k-1$ and let $\OPT$ denote the optimal solution. For every cell $C$, we have $$|D_{i,j}(C)| \leq 4~\left|\OPT \cap C_2^+\right|.$$ Consequently, by summing over all grid cells, 
$$|D_{i,j}| = \sum_C |D_{i,j}(C)| \leq 4 \sum_C \left|\OPT \cap C_2^+\right|.$$ 

We can now bound 
$\E[|D_{i,j}|]$, since
$$
\frac{\E[|D_{i,j}|]}{4} \leq ~\E\left[\sum_C \left|\OPT \cap C_2^+\right|\right] = 
\E\left[\sum_{p \in \OPT}|\CC_2^+(p)|\right] =
\sum_{p \in \OPT}\E\left[|\CC_2^+(p)|\right]
$$
by the linearity of expectation.
Note that since the expected size of $\CC_2^+(p)$ is the same for all points, it corresponds 
to the ratio between the areas of $C_2^+$ and $C$,
namely 
$$
\E\left[|\CC_2^+(p)|\right] = \frac{\mathrm{area}(C_2^+)}{\mathrm{area}(C)} = \left(\frac{k+2}{k}\right)^2.
$$
Therefore, by using inequality~(\ref{eq:k2}), we obtain
$$
\E[|D_{i,j}|] \leq 4\left(\frac{k+2}{k}\right)^2 |\OPT| \leq 4\left(1+\frac{\eps}{4}\right) |\OPT| = (4+\eps)~|\OPT|.
$$
Since the smallest among the dominating sets $D_{i,j}$ has no more than their average number of elements, the set $D^*$ returned by the algorithm satisfies 
$$
|D^*| \leq \E[|D_{i,j}|] \leq (4+\eps)~|\OPT|,
$$ closing the proof.
\end{proof}

The DS is the special case of the PDS in which $P' = P$, and thus it can be $(4+\eps)$-approximated in linear time by the same algorithm.

In the \emph{minimum distance $d$-dominating set problem} (MDDS), the input consists of a graph and an integer $d$, and the goal is to find a minimum subset of vertices such that all graph vertices are within distance at most $d$ from a vertex in the subset (here $d$ is the graph distance, that is, the number of edges in a shortest path). The DS is a special case in which $d=1$. Our algorithm for the DS can be generalized to $(4+\eps)$-approximate the MDDS in linear time for constant $d$. In contrast, the greedy heuristic that gives a $5$-approximation to the DS gives a $\Theta(d^2)$-approximation to the MDDS.

\section{VC for UDGs} \label{s:VC}

In this section, we show how to obtain a linear-time approximation scheme to the VC on unit disk graphs. We start by presenting an approximation scheme for point sets of constant diameter, and then we use the shifting strategy to generalize the result to arbitrary diameter. Differently from the two previous problems, the size of a minimum vertex cover for a point set of constant diameter is not bounded by a constant. Therefore, strictly speaking, a coreset for the problem does not exist. Nevertheless, it is possible to use coresets to approach the problem indirectly.
We remark that a different linear-time approximation scheme to the VC was presented by Marx~\cite{marx05}.

Given a graph $G=(V,E)$ with $n$ vertices, it is well known that $I$ is an independent set if and only if $V \setminus I$ is a vertex cover. While a maximum independent set corresponds to a minimum vertex cover, a constant approximation to the maximum independent set does not necessarily correspond to a constant approximation to the minimum vertex cover. However, in certain cases, an even stronger correspondence holds, as we show in the following proof.

\begin{thm} \label{t:constdiam-VC}
Given a set $P$ of $n$ points as input, with $\diam(P) = O(1)$, the VC can be $(1+\eps)$-approximated in $O(n)$ time on the real-RAM, for constant $\eps > 0$.
\end{thm}
\begin{proof}
Our algorithm considers two cases, depending on the value of $n$. If
$$n < \left(1+\frac{3}{4\eps}\right)\frac{(\diam(P)+2)^2}{4},$$
then $n$ is constant, and the VC can be solved optimally in constant time.

Otherwise, use Theorem~\ref{t:constdiam-WIS} to obtain a $4$-{ap\-prox\-i\-ma\-tion} $I$ to the maximum independent set (other constant factor approximations can be used, adjusting the threshold accordingly). We now show that $V = P \setminus I$ is a $(1+\eps)$-approximation to the minimum vertex cover. Let $I_{\OPT}, V_{\OPT}$ respectively be the maximum independent set and the minimum vertex cover. Note that $|V| = n - |I|$ and $|V_{\OPT}| = n - |I_{\OPT}|$. By a simple packing argument, dividing the area of a disk of diameter $\diam(P)+2$ by the area of a unit disk,
$$|I_{\OPT}| \leq \frac{(\diam(P)+2)^2}{4},$$
and consequently
$$n \geq \left(1+\frac{3}{4\eps}\right)|I_{\OPT}| = \left(1+\frac{3}{4\eps}\right)(n-|V_{\OPT}|).$$
Manipulating the previous inequality, we obtain
\begin{equation}\label{eq:n1}
n \leq \frac{4\eps+3}{3} ~ |V_{\OPT}|.
\end{equation}
Since $I$ is a $4$-approximation to $I_{\OPT}$,
\begin{equation}\label{eq:n2}
|V| = n - |I| \leq n - \frac{|I_{\OPT}|}{4} = \frac{4n - |I_{\OPT}|}{4} = \frac{3n + |V_{\OPT}|}{4}.
\end{equation}
Combining (\ref{eq:n1}) and (\ref{eq:n2}), we can write 
$|V| \leq (1+\eps)|V_{\OPT}|,$
as desired. 
\end{proof}

Using the shifting strategy the following result ensues.

\begin{thm} \label{t:VC}
Given a set $P$ of $n$ points in the plane as input, the VC for the corresponding UDG can be $(1+\eps)$-approximated in $O(n)$ time on the real-RAM with constant-time hashing and the floor function, for constant $\eps > 0$. Without these operations, it can be done in $O(n \log n)$ time.
\end{thm}
\begin{proof}
Let $k$ be the smallest integer such that 
\begin{equation}\label{eq:c2}
\left(\frac{k+2}{k}\right)^2 \leq \frac{1+\eps}{1+\frac{\eps}{2}}.
\end{equation} 

The algorithm proceeds as follows. For $i,j$ from $0$ to $k-1$, create a grid with cells of side $2k$ rooted at $(2i,2j)$ and, for each cell $C$ in the grid, use Theorem~\ref{t:constdiam-VC} to $(1+\eps/2)$-approximate the VC for $P \cap C_2^+$, obtaining a solution $V_{i,j}(C)$. The vertex cover $V_{i,j}$ is constructed as the union of the vertex covers $V_{i,j}(C)$ for all grid cells $C$. Return the smallest vertex cover $V_{i,j}$ that is found, call it $V^*$.

To prove that the returned solution is indeed a $(1+\eps)$-{ap\-prox\-i\-ma\-tion}, once again we use a probabilistic argument. Let $i,j$ be picked uniformly at random from $0,\ldots,k-1$ and let $\OPT$ denote the optimal solution. For every cell $C$, we have
$$|V_{i,j}(C)| \leq \left(1+\frac{\eps}{2}\right)~\left|\OPT \cap C_2^+\right|.$$
Consequently, by summing over all grid cells,
$$|V_{i,j}| = \sum_C |V_{i,j}(C)| \leq \left(1+\frac{\eps}{2}\right) \sum_C \left|\OPT \cap C_2^+\right|.$$

We can now bound 
$\E[|V_{i,j}|]$, since 
$$
\frac{\E[|V_{i,j}|]}{1+\frac{\eps}{2}} \leq ~\E\left[\sum_C \left|\OPT \cap C_2^+\right|\right] = 
\E\left[\sum_{p \in \OPT}|\CC_2^+(p)|\right] =
\sum_{p \in \OPT}\E\left[|\CC_2^+(p)|\right]
$$
by the linearity of expectation.
Note that since the expected size of $\CC_2^+(p)$ is the same for all points, it corresponds to the ratio between the areas of $C_2^+$ and $C$, namely 
$$
\E\left[|\CC_2^+(p)|\right] = \frac{\mathrm{area}(C_2^+)}{\mathrm{area}(C)} = \left(\frac{k+2}{k}\right)^2.
$$
Therefore, by using inequality~(\ref{eq:c2}), we obtain
$$
\E[|V_{i,j}|] \leq \left(1+\frac{\eps}{2}\right) \left(\frac{k+2}{k}\right)^2 |\OPT| \leq (1+\eps)~|\OPT|.
$$
Since the smallest among the vertex covers $V_{i,j}$ has no more than their average number of elements, the set $V^*$ returned by the algorithm satisfies 
$$|V^*| \leq \E[|V_{i,j}|] \leq (1+\eps)~|\OPT|,$$
closing the proof.
\end{proof}

\section{WIS for Rectangles of Bounded Size}

In this section, we consider the case in which the input is no longer a set of points, but a set of rectangles instead. Let $\lambda$ be a constant and $S$ a set of axis-aligned rectangles $R_1,\ldots,R_n$ in the plane, such that each rectangle $R_q$, for $q = 1,\ldots,n$, has width and height between $1$ and $\lambda$, and weight $w(R_q)$. Let $G$ be the intersection graph of $S$. The shifting coresets method is applied to obtain a linear-time $(6+\eps)$-approximation algorithm to the maximum-weight independent set of $G$. 

The \emph{overlap} of two rectangles $R_q,R_s$ is defined as the minimum (horizontal or vertical) translation distance of a single rectangle $R_q$ necessary to make the interiors of $R_q$ and $R_s$ disjoint. The following lemma bounds the chromatic number of the intersection graph of rectangles with a small overlap.

\begin{lem} \label{l:rect}
If $S$ is a set of axis-aligned rectangles such that
\begin{enumerate}
 \item[(1)] the width and height of each rectangle is at least $1$, and
 \item[(2)] $overlap(R_q,R_s) < 1/3$ for every two distinct rectangles $R_q,R_s \in S$,
\end{enumerate}
then the intersection graph $G$ of $S$ is $6$-colorable.
\end{lem}
\begin{proof}
Let $S$ be a set as required and $G$ its intersection graph. A graph is \emph{$1$-planar} if it can be drawn on the plane in a way that each edge intersects at most one other edge.  Borodin showed that $1$-planar graphs are $6$-colorable~\cite{Borodin}. We prove the lemma by providing such a $1$-planar drawing for $G$.

For each rectangle $R_q \in S$, draw vertex $v_q$ on the center of $R_q$. Given two intersecting rectangles $R_q,R_s$, the edge $v_qv_s$ is drawn as two straight line segments, connecting $v_q$ to the center of the rectangle $R_q \cap R_s$ and then to $v_s$.  We show that at most one other edge may cross the edge $v_qv_s$. Note that the edge $v_qv_s$ is completely inside the region $R_q \cup R_s$. 

When two rectangles $R_q, R_s \in S$ intersect one another, there are two possible types of edges corresponding to the relative positions of the rectangles (Fig.~\ref{f:rect-2cases}):
\begin{enumerate}
\item $R_q$ contains two corners of $R_s$ (or vice versa); or
\item $R_q$ contains one corner of $R_s$, and vice versa.
\end{enumerate}

\begin{figure}
 \centering
 \includegraphics[scale=.45]{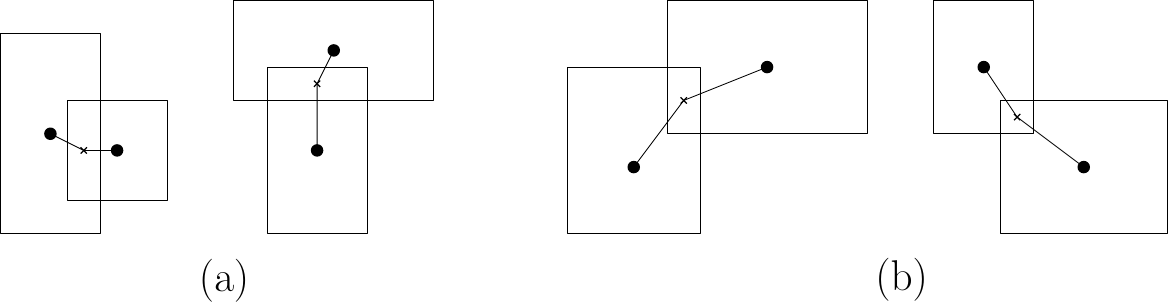}
 \caption{Examples of (a) type-1 and (b) type-2 edges.}
 \label{f:rect-2cases}
\end{figure}

We show that a type-$1$ edge $v_qv_s$ cannot possibly be crossed by any other edge. Indeed, if $R_q$ contains two corners of $R_s$, then the straight segment from $v_s$ to the center of $R_q \cap R_s$ belongs to an axis-aligned line $\ell$ that bisects $R_s$. This segment cannot be crossed by any edge $v_{q'}v_{s'}$. Otherwise, the centers of the intersecting rectangles $R_{q'}, R_{s'} \in S$ would belong to distinct halfplanes defined by $\ell$, and at least one of these rectangles, say $R_{q'}$, should cross~$\ell$. Since, by the maximum overlap allowed in $S$, the length of the segment of $\ell$ which can be contained in $R_{q'}$ measures at most $2/3$ (i.e., $1/3$ inside $R_q$ plus $1/3$ inside $R_s$), a contradiction ensues, because $R_{q'}$ measures at least $1$ on both dimensions. The straight segment from $v_r$ to the center of $R_q \cap R_s$, on its turn, cannot be crossed by any other edge $v_{q'}v_{s'}$ because, since the edge $v_{q'}v_{s'}$ is drawn completely inside the region $R_q \cup R_s$, one of the rectangles, say $R_{q'}$, should intersect that segment, yet it should not contain any of its points whose distance to the border of $R_q$ that intersects $R_s$ is greater than $1/3$. But now, since both sides of $R_{q'}$ are greater than $1$, it follows that $overlap(R_s,R_{q'}) \geq 2/3$, a contradiction.

It remains to show that a type-$2$ edge $v_qv_s$ can be crossed by at most one other edge (which must also be a type-$2$ edge). Suppose, for the sake of contradiction, that $v_qv_s$ is crossed by two other edges, $v_{q'}v_{q''}$ and $v_{s'}v_{s''}$, as illustrated in Fig.~\ref{f:rect-type2}. Now, without loss of generality, the maximum allowed overlap implies that:
\begin{figure}
 \centering
 \includegraphics[scale=.45]{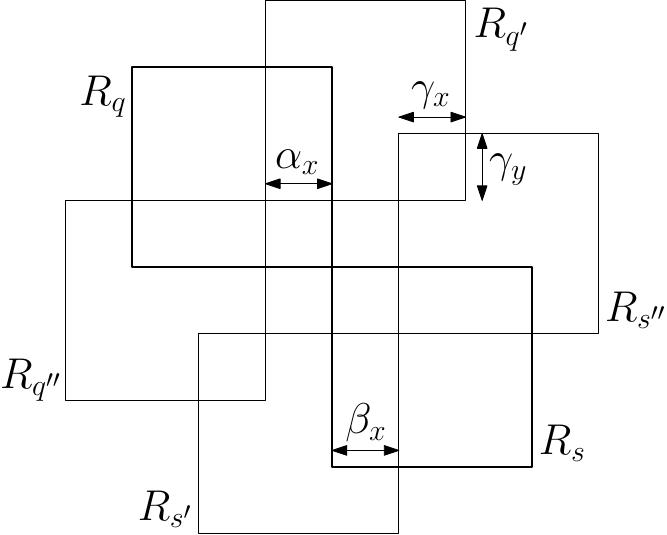}
 \caption{Crossing twice a type-2 edge.}
 \label{f:rect-type2}
\end{figure}
\begin{itemize}
\item the width $\alpha_x$ of $R_q \cap R_{q'}$ satisfies $\alpha_x < 1/3$;
\item the width $\beta_x$ of $R_s \cap R_{s'}$ satisfies $\beta_x < 1/3$;
\item the width $\gamma_x$ and height $\gamma_y$ of $R_{q'} \cap R_{s''}$ satisfy either $\gamma_x < 1/3$ or $\gamma_y < 1/3$.
\end{itemize}
If $\gamma_x < 1/3$, then the width of $R_{q'}$ is $\alpha_x + \beta_x + \gamma_x < 1$, a contradiction. If, on the other hand, $\gamma_y < 1/3$, then an analogous contradiction ensues on the $y$ direction.
\end{proof}

Given a set $S$ of rectangles, the \emph{diameter} $\diam(S)$ is the maximum distance between two vertices of the rectangles in $S$.
We are now able to prove the following theorem, which presents our $6$-{ap\-prox\-i\-ma\-tion} algorithm for sets of constant diameter.

\begin{thm} \label{t:constdiam-rect}
Given a set $S$ of $n$ axis-aligned weighted rectangles as input, such that $\diam(S) = O(1)$ and each rectangle in $S$ has width and height at least $1$, the WIS on the intersection graph of $S$ can be $6$-approximated in $O(n)$ time on the real-RAM.
\end{thm}
\begin{proof}
Represent each rectangle $R_q \in S$ by four real values $(x_q,y_q,w_q,h_q)$, corresponding 
to the $x$ and $y$ coordinates of its center, its width, and its height. The set $S$ can thus be seen as a constant-diameter set of points in $\mathbb{R}^4$. Create a $4$-dimensional grid with cells of side $\delta = 0{.}1$ (any positive $\delta < 1/9$ suffices), and define the set $S'$ by choosing the element of $S$ with maximum weight inside each non-empty grid cell. Note that since $\diam(S)=O(1)$, we have $|S'| = O(1)$, so the maximum-weight independent set among the points of $S'$ can be computed in constant-time by brute force. Return such a set.

To prove the returned solution is indeed a $6$-approximation, we show that, given an independent set $I \subseteq S$, there is an independent set $I' \subseteq S'$ such that $|I'| \geq |I|/6$. Let $J' \subseteq S'$ be the set of rectangles obtained by selecting, for each rectangle $R_i \in I$, the rectangle $R'_i \in S'$ whose corresponding $4$-dimensional point lies inside the same grid cell as that of $R_i$. Note that since $S'$ contains the maximum weight rectangle inside each grid cell, we have $w(J') \geq  w(I)$, even though $J'$ may not be an independent set. We claim that the rectangles in $J'$ overlap by less than $1/3$, hence Lemma~\ref{l:rect} can be employed to partition $J'$ into $6$ independent sets. Let $I'$ be set of maximum weight among the partitions. Since $6 w(I') \geq w(J')$, it follows that $6 w(I') \geq w(I)$, proving the theorem.

To prove the claim, consider two disjoint rectangles $R_1 = (x_1,y_1,w_1,h_1)$, $R_2 = (x_2,y_2,w_2,h_2) \in I$. Let $R'_1 = (x'_1,y'_1,w'_1,h'_1) , R'_2 = (x'_2,y'_2,w'_2,h'_2)$ be the corresponding rectangles in $J'$. Because the grid cells have side $\delta < 1/9$, we have that all the following quantities are less than $1/9$, for $i=1,2$: $|x_i - x'_i|, |y_i - y'_i|, |w_i - w'_i|, |h_i - h'_i|$. The possible horizontal overlap between $R'_1$ and $R'_2$ may come from two sources: a displacement (by less than $1/9$ for each rectangle) and a change in size, which moves the boundary of each rectangle by less than $1/18$. Therefore, the maximum horizontal overlap is at most $2/9 + 2/18 = 1/3$. The same argument bounds the maximum vertical overlap.
\end{proof}

Using the shifting strategy, we extend the result for sets of arbitrary diameter.

\begin{thm} \label{t:rect}
Let $\lambda \geq 1$ be a constant. Given a set $S$ of $n$ axis-aligned weighted rectangles as input, such that each rectangle in $S$ has width and height between $1$ and $\lambda$, the WIS can be $(6+\eps)$-approximated in $O(n)$ time on the real-RAM with constant-time hashing and the floor function. Without these operations, it can be done in $O(n \log n)$ time.
\end{thm}
\begin{proof}
Let $k$ be the smallest multiple of $\lambda$ such that
\begin{equation}\label{eq:klambda}
\left(\frac{k-\lambda}{k}\right)^2 \geq \frac{6}{6+\eps},
\end{equation} 
and let $k' = k / \lambda$.

Throughout this proof, consider grids with square cells of side $k$ and use the same strategy as in the proof of Theorem~\ref{t:WIS}, with some small modifications. To make sure that the union of two independent sets, each belonging to a different contraction of a cell, is still an independent set, we employ the contraction $C_{\lambda/2}^-$ of $C$.

For $i,j$ from $0$ to $k'-1$, create a grid with cells of side $k$ rooted at $(\lambda i, \lambda j)$. For each cell $C$ in the grid, use Theorem~\ref{t:constdiam-rect} to obtain a $6$-approximation $I_{i,j}(C)$ for the WIS whose input consists of all rectangles whose centers belong to $C$. The independent set $I_{i,j}$ is the union of all $I_{i,j}(C)$. Return the maximum-weight set $I_{i,j}$ that is found, call it $I^*$.

We now prove that the returned solution $I^*$ is indeed a $(6+\eps)$-{ap\-prox\-i\-ma\-tion}. Let $i,j$ be picked uniformly at random from $0,\ldots,k-1$ and let $\OPT$ denote the optimal solution. For every cell $C$, we have
$$w(I_{i,j}(C)) \geq \frac{1}{6}~w(\OPT \cap C_{\lambda/2}^-).$$
Consequently, by summing over all grid cells,
$$w(I_{i,j}) = \sum_C w(I_{i,j}(C)) \geq \frac{1}{6} \sum_C w(\OPT \cap C_{\lambda/2}^-)
= \frac{1}{6} \sum_{R_q \in \OPT}\rho(R_q)~w(R_q),$$
where $\rho(R_q)$ denotes the probability that the center of a given rectangle $R_q$ is contained in some contracted cell. Because such probability is the same for all rectangles, $\E[w(I_{i,j})]$ can be bounded by
$$
\E[w(I_{i,j})] \geq \frac{1}{6}~\rho(R_q)\sum_{R_q \in \OPT}w(R_q)
                 = \frac{1}{6}~\rho(R_q)~w(\OPT).
$$
Note that, for all $R_q \in S$,
$\rho(R_q)$ corresponds to the ratio between the areas of $C_{\lambda/2}^-$ and $C$, namely 
$$
\rho(R_q) = \frac{\mathrm{area}(C_{\lambda/2}^-)}{\mathrm{area}(C)} = \left(\frac{k-\lambda}{k}\right)^2.$$
Thus, inequality~(\ref{eq:klambda}) yields
$$
\E[w(I_{i,j})] \geq \frac{1}{6}\left(\frac{k-\lambda}{k}\right)^2w(\OPT) \geq \frac{1}{6+\eps}~w(\OPT).
$$

Since $I^*$ has maximum weight among the independent sets $I_{i,j}$, it follows that $w(I^*)$ is at least as large as their average weight. Therefore, 
$I^*$ 
satisfies 
$$
w(I^*) \geq \E[w(I_{i,j})] \geq \frac{1}{6+\eps}~w(\OPT),
$$ 
closing the proof.
\end{proof}

\section{Conclusion} \label{s:conclusion}

This paper introduced the method of shifting coresets, which, combining the shifting strategy and the coresets paradigm, has allowed us to obtain improved linear-time approximations for problems on unit disk graphs. The method is applicable to other geometric intersection graph classes as well.
The central idea of the method is the creation of coresets to obtain approximate solutions when the inputs are point sets of constant diameter. 
For the WIS and the DS on UDGs, the proposed algorithms provide improved approximation ratios when compared to existing linear-time algorithms, as shown in Table~\ref{t:table} in the Introduction.

\begin{figure}
 \centering
 \includegraphics[scale=.45]{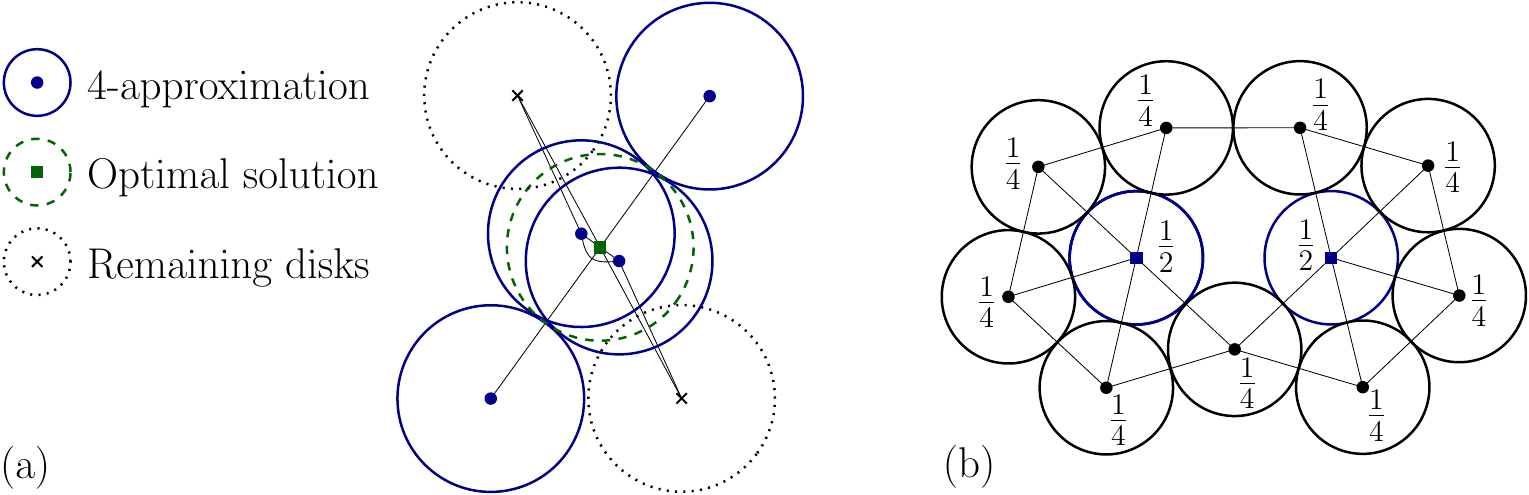}
 \caption{(a) Example in which the approximation ratio for the DS is exactly $4$. (b) Coin graph used in the example in which the approximation ratio for the WIS is $3{.}25$.}
 \label{f:lowerbound}
\end{figure}

While the approximation ratio for the WIS and the DS on UDGs is no greater than~$4$ (for constant-diameter inputs), we only know that the analysis is tight for the DS. Indeed, Fig.~\ref{f:lowerbound}(a) shows a DS instance in which our algorithm does not achieve an approximation ratio better than~$4$, even if we reduce the grid size and search for extreme points in a larger number of directions. 
In contrast, for the WIS, the best lower bound we are aware of is $3{.}25$, as shown in the following example. Let $P_1$ be the weighted point set from Fig.~\ref{f:lowerbound}(b), in which all adjacent vertices are at distance exactly $2$. Create another set $P_2$ by multiplying the coordinates of the points in $P_1$ by $1+\eps$, while multiplying their weights by $1-\eps$, for arbitrarily small $\eps > 0$. The set $P_2$ forms an independent set of weight just smaller than $3{.}25$, while the maximum independent set in $P_1$ has weight $1$. Since each vertex in $P_2$ has a smaller weight and is arbitrarily close to a vertex of $P_1$, the vertices of $P_2$ will be disregarded by the algorithm for the input instance $P_1 \cup P_2$.

The analysis of the $6$-approximation ratio for the WIS on rectangle intersection graphs leaves an even bigger gap. The best lower bound we are aware of is $13/3$, since the graph illustrated in Fig.~\ref{f:rect-graph} (with $13$ vertices and maximum independent set with size $3$) is in the graph class used in Lemma~\ref{l:rect}. In fact, it is possible that such class is $5$-colorable (the same graph in Fig.~\ref{f:rect-graph} shows it is not $4$-colorable, though). We remark that the need for $c$ colors does not mean that the ratio between the total weight of the vertices and the maximum weight of an independent set can be as high as $c$.

\begin{figure}
 \centering
 \includegraphics[scale=.45]{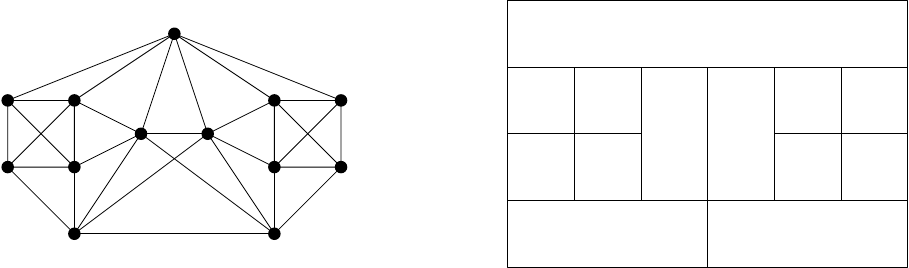}
 \caption{A $5$-chromatic graph in the class from Lemma~\ref{l:rect}, and its representation by slightly-overlapping rectangles.}
 \label{f:rect-graph}
\end{figure}

Several open problems remain. Is it possible to obtain an approximation ratio better than $4$ in $O(n \log n)$ time for the WIS on UDGs, or at least for the unweighted version? Can the linear-time approximation scheme for VC be generalized for the weighted version? Are the point coordinates really necessary, or is it possible to devise similar graph-based algorithms? Also, can our method be used to obtain better linear-time approximations to related problems on unit disk graphs such as finding the minimum-weight dominating set or the minimum connected dominating set? Is it possible to use the shifting coresets method to obtain a constant approximation for the WIS on disk graphs (of arbitrary radii) in linear time? Finally, is it possible to use similar ideas to derive improved linear-time approximations for problems on other classes of graphs such as planar graphs, bounded treewidth graphs, etc.?

\section*{Acknowledgements}

The authors would like to thank Raphael Machado, Mickael Montassier, Petru \mbox{Valicov} and Yann Vax\`es for the insightful discussions. The proofs of Theorems~\ref{t:WIS} and~\ref{t:VC} use the same techniques as in Hunt III \emph{et al}., which were in turn motivated by Baker's technique~\cite{ptas-geometric,baker94}.

This research was partially supported by the Brazilian agencies CAPES, CNPq, and FAPERJ. An extended abstract of this paper appeared in the 12th Workshop on Approximation and Online Algorithms (WAOA 2014).

\bibliographystyle{plain}
\bibliography{udg}

\end{document}